\documentclass{amsart}

\usepackage{amssymb, amsmath, amscd, amsthm}
\usepackage{graphicx}
\usepackage{psfrag}
\usepackage[all]{xy}
\usepackage{svg}

\newtheorem{thm}{Theorem}
\newtheorem{lem}{Lemma}

\theoremstyle{definition}

\theoremstyle{remark}

\newcommand{\R}{\mathbb R}
\newcommand{\F}{\mathbb F}
\newcommand{\N}{\mathbb N}
\numberwithin{equation}{section}

\newcommand{\p}{\varphi}
\newcommand{\eps}{\varepsilon}

\begin{document}

\title{The rank invariant stability via interleavings}

\author{Claudia Landi}
\address{Dipartimento
di Scienze e Metodi dell'Ingegneria, Universit\`a di Modena e
Reggio Emilia, Via Amendola 2, Pad. Morselli, I-42100 Reggio
Emilia, Italia\newline ARCES, Universit\`a di Bologna, via Toffano
$2/2$, I-$40125$ Bologna, Italia} \email{claudia.landi@unimore.it}


%
%

\begin{abstract}
A lower bound for the interleaving distance on persistence vector spaces is given in terms of rank invariants. This offers an alternative proof of the stability of rank invariants.
\end{abstract}

\maketitle

\section*{Introduction}
Increasingly often in recent years the shape of data has   been analyzed using  persistent homology \cite{Carlsson2009}.   Given a topological space (e.g., a simplicial complex built upon a finite set of points in $\R^n$), one  constructs a nested family of subspaces called a
filtration, and studies the topological events occurring along the filtration. These are encoded in a structure  known as a persistence vector space, and consists of a family of vector spaces obtained as  the homology
of the  subspaces in the filtration, connected by linear maps  induced in homology by the inclusion maps of the filtration. By constructing
a filtration and taking its  homology one transforms a topological problem into linear algebra. This theory  is well understood when filtrations only depend on one parameter, whereas still presents many open problems  in the  multi-parameter case, commonly referred to  as the multidimensional persistent homology theory. On the other hand, applications strongly motivate the interest in multi-filtrations.

Recently, the problem of comparing  persistence vector spaces in a stable and optimal way has been successfully solved by using families of linear maps between  persistence vector spaces, known as interleavings \cite{LesnickXX}. However, in concrete cases, computing the interleaving distance is still not a viable option because of its complexity. Therefore it may be useful to find estimates for it.

The rank invariant  is the most studied invariant of persistence vector spaces. The primary goal of this paper is to show that the multidimensional matching distance  on rank invariants studied in \cite{Cerri2013}  provides a lower bound for the interleaving distance. From a different standpoint,  this fact can be viewed as a new proof that the rank invariant is a stable invariant when rank invariants are compared via 1-dimensional reduction along lines. Indeed, the secondary goal of this paper is to obtain  new proofs of the rank invariant stability and internal stability results using interleavings.      

\section{Background defitions}
Let $n\in \N_{>0}$. For every $u=(u_1,\dots,u_n),v=(v_1,\dots,v_n)\in\R^n$, we write $u\preceq v$ (resp. $u\prec v$, $u\succ v$, $u\succeq v$) if and only if $u_i\leq v_i$ (resp. $u_i<v_i$, $u_i>v_i$, $u_i\geq v_i$) for  $i=1,\dots,n$. Note that $u\succ v$ is not the negation of $u\preceq v$. 

When $\R^n$ is viewed as a vector space, its elements are denoted using overarrows.  Moreover, in this case, we endow $\R^n$ with the max-norm defined by $\|\vec v\|_{\infty}=\max_i |v_i|$.

For a field $\F$, an $n$-dimensional persistence vector space $\bf M$ is a family $\{{\bf M}_u\}_{u\in \R^n}$ of  $\F$-vector spaces, together with a family of linear maps $\{\p_{\bf M}(u,v): {\bf M}_u \to {\bf M}_v\}_{u\preceq v}$ such that, for all $u\preceq v\preceq w\in \R^n$, $\p_{\bf M}(u,u)=\mathrm{Id}_{{\bf M}_u}$, and $\p_{\bf M}(v,w)\circ \p_{\bf M}(u,v)=\p_{\bf M}(u,w)$.
We call the maps $\p_{\bf M}(u,v)$ transition maps. {\bf M} is said to be {\em pointwise finite dimensional},  if $\dim({\bf M}_u) < \infty$ for all $u \in \R^n$.

A morphism $\alpha:{\bf M}\to{\bf N}$  of persistence vector spaces is  a collection of linear maps $\alpha(u) : {\bf M}_u \to {\bf N}_u$ such that $\alpha(v)\circ \p_{\bf M}(u,v)=\p_{\bf N}(u,v)\circ \alpha(u)$, for all $u\preceq v\in \R^n$.

For {\bf M} a persistence vector space, and for $\eps \ge 0$, ${\bf M}(\vec \eps)$ denotes the module ${\bf M}$ diagonally shifted by $\vec \eps= (\eps, \eps,\ldots, \eps)$:  ${\bf M}(\vec \eps)_u = {\bf M}_{u+\vec \eps}$,
and for $u\preceq v\in\R^n$, $\p_{{\bf M}(\vec \eps)}(u, v) = \p_{\bf M}(u + \vec \eps, v + \vec \eps)$. We also  let ${\bf M}(\vec \eps) : {\bf M} \to {\bf M}(\vec \eps)$ be the diagonal $\eps$-transition
morphism, that is  the morphism whose restriction to ${\bf M}_u$ is the linear map $\p_{\bf M}(u, u +\vec\eps)$ for all $u\in\R^n$. We say that two $n$-modules ${\bf M}$ and ${\bf N}$ are $\eps$-interleaved if there exist morphisms
$\alpha: {\bf M} \to {\bf N}(\vec\eps)$ and $\beta : {\bf N}\to  {\bf M}(\vec\eps)$ such that $\beta(\vec\eps) \circ \alpha = \p_{\bf M}(2\vec\eps)$ and
$\alpha(\vec\eps) \circ \beta = \p_{\bf N}(2\vec\eps)$. The {\em interleaving distance} on persistence vector spaces is defined  by setting
$$d_I({\bf M},{\bf N}) = \inf \{\eps\in[0,+\infty): \mbox{${\bf M}$ and ${\bf N}$ are $\eps$-interleaved}\}.$$

In dimension $n=1$ the theory is well understood and nicely reviewed in \cite{Lesnick2014}.  In particular,  any pointwise finite dimensional persistence
module  is completely representable by a   unique   multiset of intervals  $B({\bf M})$, called a {\em barcode} (or equivalently, a multiset of points of $\R^2$, called a {\em persistence diagram}).  The {\em bottleneck distance}  $d_B$ is equal to the interleaving distance on barcodes, and the  Algebraic Stability of Persistence Barcodes states that for any two pointwise finite dimensional persistence vector spaces ${\bf M}$ and ${\bf N}$ of dimension 1 it holds that $d_B(B({\bf M}), B({\bf N})) \le d_I({\bf M}, {\bf N})$.

In dimension $n>1$ the persistence vector space structure is still matter of investigation, and numeric invariants are often used instead. The {\em rank invariant} of persistence vector spaces is  defined by setting $\rho_{\bf M}(u,v)=\mathrm{rank}\p_{\bf M}(u,v)$ for every $u\preceq v\in \R^n$. In \cite{Cerri2013} a  readily computable
metric on rank invariants, the multidimensional matching distance $d_{match}$, is defined  via  1-dimensional reductions. 

In the next section the connections between the rank invariant and interleavings are highlighted.

\section{The rank invariant and interleavings}
 
Given a line $L$ in $\R^n$  parameterized by $u=s\vec m+b$, with $m^*=\min_i m_i>0$, we denote by ${\bf M}_L$ the persistence vector space parameterized by $s\in\R$ and obtained by  restriction of ${\bf M}$ to $L$: $({\bf M}_L)_s={\bf M}_u$ and $\p_{{\bf M}_L}(s,s')=\p_{\bf M}(u,u')$ for $u=s\vec m+b$ and $u'=s'\vec m+b$.

\begin{lem}\label{inter}
 If ${\bf M}$ and ${\bf N}$ are $\eps$-interleaved then  ${\bf M}_L$ and ${\bf N}_L$ are $\frac{\eps}{m^*}$-interleaved.
\end{lem}

\begin{proof}
Because ${\bf M}$ and ${\bf N}$ are $\eps$-interleaved, there exist two morphisms $\alpha:{\bf M} \to {\bf N}(\vec \eps)$ and $\beta:{\bf N} \to {\bf M}(\vec \eps)$ such that $\beta(u+\vec\eps)\circ \alpha(u)=\p_{\bf M}(u,u+2\vec\eps)$ and $\alpha(u+\vec\eps)\circ \beta(u)=\p_{\bf N}(u,u+2\vec\eps)$, for every $u\in\R^n$.

We need to prove that  there are two morphisms $\alpha_L:{\bf M}_L\to {\bf N}_L(\frac{\eps}{m^*})$ and $\beta_L: {\bf N}_L\to {\bf M}_L(\frac{\eps}{m^*})$ such that $\beta_L\left(s+\frac{\eps}{m^*}\right)\circ \alpha_L(s)=\p_{{\bf M}_L}(s,s+2\frac{\eps}{m^*})$ and $\alpha_L(s+\frac{\eps}{m^*})\circ \beta_L(s)=\p_{{\bf N}_L}(s,s+2\frac{\eps}{m^*})$, for every $s\in\R$. The idea underlying the construction of $\alpha_L$ and $\beta_L$ is illustrated in Figure~\ref{fig:stability}.

Let $\hat\iota=\mathrm{argmin} _i m_i$. For any $\bar u\in L$,   we set $\bar u'$ as the point of intersection between $L$ and the hyperplane $\pi': u_{\hat\iota}=\bar u_{\hat\iota}+\eps$ containing $\bar u+\vec \eps$. A direct computation shows that $\bar u'=\bar u+\frac{\eps}{m^*}\cdot \vec m$. Hence $\bar u\prec \bar u'$ and $\bar u'=s'\vec m+b$ with $s'=s+\frac{\eps}{m^*}$. Analogously,  $u''=\left (s+2\frac{\eps}{m^*}\right)\vec m+b$ turns out to be the point of intersection between $L$ and the hyperplane $\pi'': u_{\hat\iota}=\bar u'_{\hat\iota}+\eps$ containing $\bar u'+\vec \eps$.

We are now ready to  construct $\alpha_L$ and $\beta_L$. Recall that $u,u'$ are functions of $s$. We define $\alpha_L: {\bf M}_L\to {\bf N}_L(\frac{\eps}{m^*})$ by setting $\alpha_L(s)=\alpha(u'-\vec\eps)\circ \p_{\bf M}(u,u'-\vec\eps)$ (red arrows in Figure~\ref{fig:stability}).  $\alpha_L$ is well defined because a direct computation shows that $u\preceq u'-\vec\eps$. We define $\beta_L: {\bf N}_L\to {\bf M}_L(\frac{\eps}{m^*})$ by setting $\beta_L(s)=\p_{\bf M}(u+\vec\eps,u')\circ\beta(u)$. Also $\beta_L$ is well defined because  $u+\vec\eps\preceq u'$ (green arrows in Figure~\ref{fig:stability}).

\begin{figure}
\includesvg[width=0.5\textwidth]{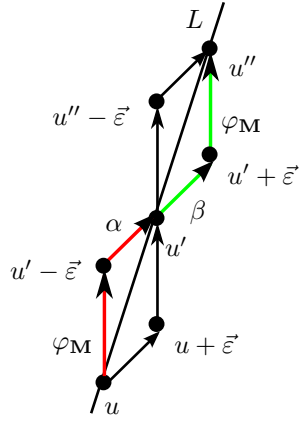}
\caption{The maps $\alpha_L$ and $\beta_L$ defined in Proposition~\ref{inter}:   $\alpha_L$ is given by the composition of red maps, $\beta_L$ is given by the composition of green maps.} \label{fig:stability}
\end{figure}

Observing that 
$$\beta_L\left(s+\frac{\eps}{m^*}\right)\circ \alpha_L(s)=\p_{\bf M}(u'+\vec\eps,u'')\circ\beta(u') \circ \alpha(u'-\vec\eps)\circ \p_{\bf M}(u,u'-\vec\eps),$$ and recalling that $\beta(u') \circ \alpha(u'-\vec\eps)= \p_{\bf M}(u'-\vec\eps,u'+\vec\eps)$, we obtain 
$$\beta_L\left(s+\frac{\eps}{m^*}\right)\circ \alpha_L(s)=\p_{\bf M}(u'+\vec\eps,u'')\circ \p_{\bf M}(u'-\vec\eps,u'+\vec\eps)\circ \p_{\bf M}(u,u'-\vec\eps)$$
immediately yielding
$$\beta_L\left(s+\frac{\eps}{m^*}\right)\circ \alpha_L(s)=\p_{\bf M}(u,u''),$$
or equivalently
$$\beta_L\left(s+\frac{\eps}{m^*}\right)\circ \alpha_L(s)=\p_{{\bf M}_L}(s,s+2\frac{\eps}{m^*}).$$

Analogously, observing that 
$$\alpha_L\left(s+\frac{\eps}{m^*}\right)\circ \beta_L(s)= \alpha(u''-\vec\eps)\circ \p_{\bf M}(u',u''-\vec\eps) \circ\p_{\bf M}(u+\vec\eps,u')\circ\beta(u)$$ and recalling that by definition of morphism of persistence modules we have 
$$\alpha(u''-\vec\eps)\circ \p_{\bf M}(u',u''-\vec\eps)=\p_{\bf N}(u'+\vec\eps,u'')\circ \alpha(u')$$ and $$\p_{\bf M}(u+\vec\eps,u')\circ\beta(u)=\beta(u'-\vec\eps)\circ\p_{\bf N}(u,u'-\vec\eps),$$ we obtain 
$$\alpha_L\left(s+\frac{\eps}{m^*}\right)\circ \beta_L(s)=\p_{\bf N}(u'+\vec\eps,u'')\circ \alpha(u')\circ \beta(u'-\vec\eps)\circ\p_{\bf N}(u,u'-\vec\eps).$$
Hence, from $\alpha(u')\circ \beta(u'-\vec\eps)=\p_{\bf N}(u'-\vec\eps,u'+\vec\eps)$, we get
$$\alpha_L\left(s+\frac{\eps}{m^*}\right)\circ \beta_L(s)=\p_{\bf N}(u'+\vec\eps,u'')\circ \p_{\bf N}(u'-\vec\eps,u'+\vec\eps)\circ\p_{\bf N}(u,u'-\vec\eps)$$
which immediately yields 
$$\alpha_L\left(s+\frac{\eps}{m^*}\right)\circ \beta_L(s)=\p_{\bf N}(u,u'')$$
or equivalently
$$\beta_L\left(s+\frac{\eps}{m^*}\right)\circ \alpha_L(s)=\p_{{\bf M}_L}(s,s+2\frac{\eps}{m^*}).$$
\end{proof}

\begin{thm}
For every pointwise finite dimensional modules $\bf M$ and $\bf N$, it holds that $d_{match}(\rho_{\bf M}, \rho_{\bf N})\le d_I({\bf M}, {\bf N})$.
\end{thm}

\begin{proof}
By definition $d_{match}(\rho_{\bf M}, \rho_{\bf N})=\sup_{L:u=s\vec m +b}m^*\cdot d_B(B({\bf M}_L), B({\bf N}_L))$ where $L$ varies in the set of all the lines parameterized by $u=s\vec m+b$, with $m^*=\min_i m_i>0$, $\max_i m_i=1$, $\sum_{i=1}^n b_i=0$.    By Lemma~\ref{inter}, for any such line $L$, if ${\bf M}$ and ${\bf N}$ are $\eps$-interleaved then ${\bf M}_L$ and ${\bf N}_L$ are $\frac{\eps}{m^*}$-interleaved. Thus, $m^*\cdot d_I({\bf M}_L, {\bf N}_L)\le  d_I({\bf M}, {\bf N})$, yielding $m^*\cdot d_B(B({\bf M}_L), B({\bf N}_L))\le  d_I({\bf M}, {\bf N})$ by the Algebraic Stability of Persistence Barcodes. Hence the claim.
\end{proof}

\bigskip

Let us now consider the problem of internal stability of rank invariants. For $L$  parameterized by $u=s\vec m+b$, with $m^*=\min_i m_i>0$,  and $L'$ parameterized by $u=s\vec m'+b'$, with $m'^*=\min_i m_i'>0$,  we denote by ${\bf M}_L$ and ${\bf M}_{L'}$ the persistence modules parameterized by $s\in \R$ obtained by  restriction of ${\bf M}$ to $L$ and $L'$, respectively, and consider the interleavings between ${\bf M}_L$ and ${\bf M}_{L'}$.

\begin{thm}\label{internal}
Let $\bf M$ be a  pointwise finite dimensional persistence vector space for which there exist $c=(c_1,c_2,\ldots,c_n)\in\R^n$ such that $\p_{\bf M}(u,u')$ is an isomorphism for every $u,u'\in \R^n$ with $c\prec u\preceq u'$.   There exist constants $K,C>0$ such that  ${\bf M}_L$ and ${\bf M}_{L'}$ are $\eta$-interleaved, and therefore $d_B({\bf M}_L,{\bf M}_{L'})\le \eta$, with $$ \eta=\frac{K\cdot \|\vec m-\vec m'\|_\infty+ C\cdot \|b-b'\|_\infty}{m^*\cdot m'^*}.$$
\end{thm}

\begin{proof}
 We set $C= \max \{\|\vec m\|_\infty,\|\vec m'\|_\infty\}$ and $K=A+2B$ with $B=\max \{\|b\|_\infty,\|b'\|_\infty\}$ and 
$$A=\max\{\max_j |c_j-b_j|\cdot \frac{\|\vec m\|_\infty}{m^*}, \max_j |c_j-b'_j|\cdot \frac{\|\vec m'\|_\infty}{m'^*} \}.$$

We now show that with this setting there exist  morphisms $\alpha:{\bf M}_L\to {\bf M}_{L'}(\eta)$ and $\beta: {\bf M}_{L'}\to {\bf M}_L(\eta)$ such that $\beta\left(s+\eta\right)\circ \alpha(s)=\p_{{\bf M}_L}(s,s+2\eta)$ and $\alpha(s+\eta)\circ \beta(s)=\p_{{\bf M}_{L'}}(s,s+2\eta)$, for every $s\in\R$. The idea of the construction of $\alpha$ and $\beta$ is illustrated in Figure~\ref{fig:internal}.

\begin{figure}
\includesvg[width=\textwidth]{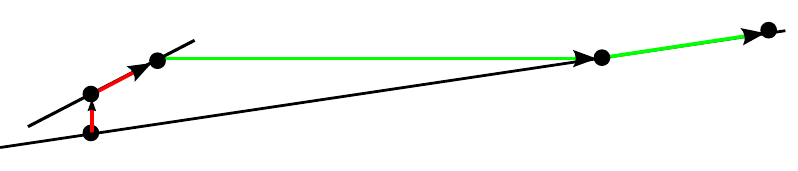}
\includesvg[width=0.95\textwidth]{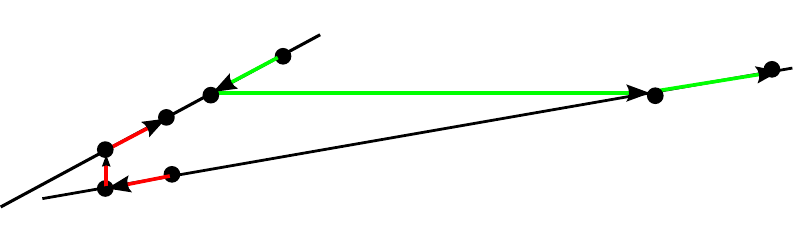}
\caption{The maps $\alpha$ and $\beta$ defined in Proposition~\ref{internal}: $\alpha$ is given by the composition of red maps, $\beta$ is given by the composition of green maps. Top: the case $s\le \max_j \frac{c_j-b_j}{m^*}$. Bottom: the case $s> \max_j \frac{c_j-b_j}{m^*}$. } \label{fig:internal}
\end{figure}

In order to construct $\alpha$, we note that for each $u=s\vec m+b$ in $L$,    there exists, and is unique,  a point $v=t\vec m'+b'$ in $L'$, with $u\preceq v$ and $  v_{\hat\iota}= u_{\hat \iota}$ for some $\hat \iota$, $1\le\hat  \iota\le n$, because  $m'^*=\min_i m_i'>0$. We also let $w=(s+\eta)\cdot \vec m'+b'$. We  consider separately the cases $s\le \max_j \frac{c_j-b_j}{m^*}$ and $s> \max_j \frac{c_j-b_j}{m^*}$. We will see that in the first case it holds that $u\preceq w$ and therefore  the morphism $\p_{\bf M}(u,w)$ is defined between $({\bf M}_L)_s$ and $({\bf M}_{L'})_{s+\eta}$; in the second case, we will see that  the transition maps are isomorphisms.
Therefore we can set, for any $s\in\R$, 
$$\alpha(s)=\left\{\begin{array}{ll}
\p_{\bf M}(u,w) & \mbox{ for $s\le \max_j \frac{c_j-b_j}{m^*}$}\\
\p_{\bf M}(\bar w, w)\circ \p_{\bf M}(\bar u,\bar w)\circ \p_{\bf M}^{-1}(\bar u,u) & \mbox{ for $s>\max_j \frac{c_j-b_j}{m^*}$}
\end{array}\right.$$
where $u=s\cdot \vec m+b$ and $w=(s+\eta)\cdot \vec m'+b'$, and, setting $\bar s=\max_j \frac{c_j-b_j}{m^*}$,  $\bar u= \bar s\cdot \vec m+b$,  $\bar w=\left(\bar s+\eta \right)\cdot \vec m'+b'$. 

Let us first see that $u\preceq w$ when $s\le \max_j \frac{c_j-b_j}{m^*}$. We have
$$ t-s= \frac{u_{\hat\iota}-b'_{\hat \iota}}{m'_{\hat\iota}}-\frac{u_{\hat\iota}-b_{\hat \iota}}{m_{\hat\iota}}= \frac{ u_{\hat\iota}(m_{\hat\iota}-m'_{\hat\iota})+b_{\hat\iota}(m'_{\hat\iota}-m_{\hat\iota})+m_{\hat\iota}(b_{\hat\iota}-b'_{\hat\iota})}{m_{\hat\iota} m'_{\hat\iota}}.$$
The assumption $s\le \max_j \frac{c_j-b_j}{m^*}$ implies that 
$$u_{\hat\iota}=s \cdot  m_{\hat\iota}+b_{\hat\iota}\le \max_j \frac{c_j-b_j}{m^*}\cdot  m_{\hat\iota}+b_{\hat\iota}\le A+B.$$
Hence
\begin{equation}\label{disug}
t\le s+ \frac{K\cdot \|\vec m-\vec m'\|_\infty+ C\cdot \|b-b'\|_\infty}{m^*\cdot m'^*}=s+\eta.
\end{equation}
Thus $v\preceq w$, yielding $u\preceq w$ for $s\le \max_j \frac{c_j-b_j}{m^*}$.

Let us now see that, when $s> \max_j \frac{c_j-b_j}{m^*}=\bar s$, we have $u\succ \bar u\succeq c$ and $w\succ \bar w\succeq  c$. Indeed, for $i=1,\ldots , n$, it holds that  
\begin{equation}\label{c}
\bar s\cdot m_i+b_i = \max_j \frac{c_j-b_j}{m^*}\cdot m_i+b_i\ge \max_j (c_j-b_j)+b_i\ge  c_i,
\end{equation}
implying that $u\succ \bar u\succeq c$. In order to see that $w\succ \bar w\succeq  c$ too, let us set  $\bar v$ equal to the unique point in $L'$ such that $\bar u\preceq \bar v$ and $  \bar v_{\hat\iota}= \bar u_{\hat \iota}$  for some $\hat \iota$, $1\le\hat  \iota\le n$. By inequality (\ref{disug}) it holds that $\bar v\preceq \bar w$. Thus, from $s>\bar s$ and inequalities (\ref{c}) it follows  $w\succ \bar w \succeq \bar v\succeq \bar u\succeq c$. 
Hence $\alpha$ is well defined. Analogously,  we can define $\beta$.

Finally, it certainly holds that $\alpha$ and $\beta$ are morphisms of persistence persistence vector spaces, and  $\beta\left(s+\eta\right)\circ \alpha(s)=\p_{{\bf M}_L}(s,s+2\eta)$ and $\alpha(s+\eta)\circ \beta(s)=\p_{{\bf M}_{L'}}(s,s+2\eta)$ because $\alpha$ and $\beta$ have been defined using $\p_{\bf M}$ itself.

Having proved that ${\bf M}_L$ and ${\bf M}_{L'}$ are $\eta$-interleaved, the inequality $d_B({\bf M}_L,{\bf M}_{L'})\le \eta$ follows from the Algebraic Stability Theorem for Barcodes.

\end{proof}

\section{Acknowledgments}
The author wishes to thank M. Lesnick for  providing  constructive critical comments on the results presented here.

\end{document}